\documentclass[11pt]{article}
\usepackage[margin=1in]{geometry}

\usepackage[pdftex]{graphicx}
\graphicspath{{../pdf/}{../jpeg/}}
\DeclareGraphicsExtensions{.pdf,.jpeg,.png}
\usepackage{amsmath}
\usepackage{algorithm}  
\usepackage{algorithmic}
\usepackage{array}
\usepackage{url}
\usepackage{textcomp}
\usepackage{hyperref}
\usepackage{amssymb,amsmath,url}
\usepackage{bm} 
\usepackage{multirow}
\usepackage{booktabs}	
\usepackage{epstopdf}
\usepackage{color}
\usepackage{subfigure}
\usepackage{pifont}
\usepackage{tikz}
\usepackage{color}
\usepackage{cite}
\usetikzlibrary{calc}
\usepackage{soul} 
\usepackage{graphics}
\newtheorem{theorem}{Theorem}
\newtheorem{lemma}{Lemma}
\newtheorem{proposition}{Proposition}

\allowdisplaybreaks
\newtheorem{definition}{Definition} 

\newtheorem{proof}{Proof}

\newtheorem{example}{Example}

\newcommand{\CC}{\mathcal{C}}
\newcommand{\BB}{\mathcal{B}}
\newcommand{\R}{\mathbb{R}}

\newcommand{\N}{\mathbb{N}}
\newcommand{\X}{\mathbb{X}}

\newcommand{\norm}[1]{\left\lVert#1\right\rVert}

\newcommand*{\QEDA}{\hfill\ensuremath{\blacksquare}}

\DeclareMathOperator*{\argmin}{arg\,min}

\DeclareMathOperator*{\argmax}{arg\,max}

\allowdisplaybreaks
\begin{document}

\title{Core-Selecting Mechanisms in Electricity Markets\let\thefootnote\relax\footnotetext{This research was gratefully funded by the European Union ERC Starting Grant CONENE.}}%
\author{Orcun Karaca\thanks{The authors are with the Automatic Control Laboratory, Department of Information Technology and Electrical Engineering, ETH Z\"{u}rich, Switzerland. e-mails: {\tt\small \{okaraca, mkamgar\}@control.ee.ethz.ch}} \and Maryam Kamgarpour\footnotemark[1]
}

\maketitle

\begin{abstract}
	Due to its theoretical virtues, several recent {works propose} the use of the incentive-compatible Vickrey-Clarke-Groves (VCG) mechanism for electricity markets. Coalitions of participants, however, can influence the VCG outcome to obtain higher collective profit. To address this issue, we propose core-selecting mechanisms for their coalition-proofness. We show that core-selecting mechanisms generalize the economic rationale of the locational marginal pricing (LMP) mechanism. Namely, these mechanisms are the exact class of mechanisms that ensure the existence of a competitive equilibrium in linear/nonlinear prices. This implies that the LMP mechanism is also core-selecting, and hence coalition-proof.  In contrast to the LMP mechanism, core-selecting mechanisms exist for a broad class of electricity markets, such as ones involving nonconvex costs and nonconvex constraint sets. In addition, they can approximate truthfulness without the price-taking assumption of the LMP mechanism. Finally, we show that they are also budget-balanced. Our results are verified with case studies based on optimal power flow test systems and the Swiss reserve market.
\end{abstract}
\vspace{.1cm}
\section{Introduction}\label{sec:1}

{A}{ rapid} transformation has been underway for the last couple of decades to replace tight regulation of the electricity industry with competitive market structures~\cite{wilson2002architecture}. This liberalization has been essential to improve economic efficiency and to attract new investments\cite{cramton2017electricity}. Designing electricity markets, however, is a complex task. One inherent complexity is the need to achieve real-time balance of supply and demand because of an inability to store electricity efficiently\cite{cramton2003electricity}. This task is made more difficult by both intertemporal and network dependencies, and more recently, by high penetration of intermittent renewable energy sources~\cite{cramton2017electricity}. Hence, there has been a surge of interest in proposing new electricity market designs~\cite{bosesome}.

This work studies electricity markets in which generators first submit bids representing their underlying economic costs, and an operator then optimizes all the resources to secure a reliable operation. The principal element of such markets is the payment made to each generator since the generators have incentives to strategize around~these payments. As such, payment rules need to be carefully designed to ensure that the generators reveal their true costs which would then yield a stable grid with maximum efficiency, in other words, maximum social welfare~\cite{cramton2017electricity}.

The locational marginal pricing (LMP) mechanism is a well-studied payment rule used in many existing electricity markets~\cite{schweppe2013spot,wu1996folk}. It is based on using the Lagrange multipliers of nodal balance equations in optimal power flow~(OPF) problems to form linear prices. If generators assume that the Lagrange multipliers are independent of their bids, then the LMP mechanism is incentive-compatible, that is, generators are incentivized to submit their true costs. This assumption, also called price-taking, arises from competitive equilibrium theory; however, it often does not hold in practice~\cite{mas1995microeconomic}. In particular, empirical evidence has shown that strategic manipulations have increased the LMP payments substantially in electricity markets~\cite{joskow2001quantitative}. Moreover, LMP payments are well-defined only under convexity assumptions on the bids and the constraints~\cite{o2005efficient}. Convexity is a simplifying abstraction of many realistic electricity market models~\cite{lavaei2012competitive,warrington2012market}. Without such restrictions, it is not possible to guarantee the existence of meaningful Lagrange multipliers~\cite{bikhchandani1997competitive}.

In contrast to the LMP mechanism, the Vickrey-Clarke-Groves (VCG) mechanism ensures that truthful bidding is the dominant-strategy Nash equilibrium~\cite{vickrey1961counterspeculation,clarke1971multipart,groves1973incentives}. Consequently, several recent works have proposed the use of this payment rule in a broad class of electricity markets~\cite{samadi2012advanced, pgs,xu2017efficient,karaca2017game,karaca2018weak}. However, the VCG mechanism is often deemed undesirable for practical applications since coalitions of generators can strategically bid to increase their collective utility. As a result, it is susceptible to different kinds of manipulations such as collusion and shill bidding~\cite{hobbs2000evaluation,ausubel2006lovely,rothkopf2007thirteen}. In particular, our recent work has shown that these manipulations are eliminated from the VCG mechanism only under convex bids and polymatroid-type constraints~\cite{karaca2017game}.\footnote{Polymatroid is a polytope associated with a submodular function~\cite{schrijver2003combinatorial}.} This restricted market problem does not capture realistic electricity market models. Moreover, since the same participants are involved in similar market transactions day after day, electricity markets can particularly be exposed to collusion and shill bidding~\cite{anderson2011implicit}.  

 These challenges motivate our work on coalition-proof mechanisms that are immune to collusion and shill bidding. These mechanisms are based on the idea of the core in coalitional game theory~\cite{osborne1994course}, and are referred to as core-selecting mechanisms. They were first proposed for auctions of multiple items in \cite{day2008core}, and then applied~to auctions of continuous goods (e.g., electrical power) in \cite{karaca2017game}. In this paper, we show that core-selecting mechanisms are the exact class of mechanisms that ensure the existence of a competitive equilibrium under linear/nonlinear prices. Our result implies that the LMP mechanism is core-selecting, and hence coalition-proof. The equivalence of competitive equilibrium and the core was first shown in~\cite{shapley1971assignment} for item exchanges under unit demand and unit supply, for example, house allocation problems. Our work extends this result to auctions involving continuous goods, second stage costs, and general market constraints.
 In contrast to the LMP mechanism, we show that core-selecting mechanisms are applicable to a broad class of electricity markets, such as the ones involving nonconvex costs and/or nonconvex constraint sets. In addition, core-selecting mechanisms can approximate incentive-compatibility without the price-taking assumption of the LMP mechanism. We highlight that the benefits of core-selecting mechanisms are accompanied by nonlinear pricing which might be regarded as a big shift for some existing markets~\cite{bosesome}.

Our contributions are as follows. First, we prove that for electricity markets any competitive equilibrium is efficient. Second, we prove that a mechanism is core-selecting if and only if it ensures the existence of a competitive equilibrium. This implies that the LMP mechanism is core-selecting. Third, we derive an upper bound on the profit a bidder can obtain by a unilateral deviation from its truthful bid, under any core-selecting mechanism. Using this bound, we propose a mechanism that maximizes incentive-compatibility among all core-selecting mechanisms. Fourth, we show that any core-selecting mechanism is budget-balanced in exchange markets. Finally, we verify our results with case studies based on real-world electricity market data.

In Section~\ref{sec:2}, we introduce a general class of electricity markets and discuss desirable properties in mechanism design. Using tools from coalitional game theory and competitive equilibrium theory, Section~\ref{sec:3} proves the equivalence of core and competitive equilibrium. Then, we investigate incentive-compatibility and budget-balance. Section~\ref{sec:5} presents the numerical results. All proofs are relegated to the appendix.

\section{Mechanism Framework}\label{sec:2}

We start with a generic one-sided electricity market reverse auction.\footnote{Our results can be generalized to exchanges, see Section~\ref{sec:4}.} The set of participants consists of the central operator $l=0$ and the bidders $L=\{1,\ldots,\lvert L\rvert\}$. Let there be $t$ types of power supplies in the auction. These types can include control reserves, also known as ancillary services~\cite{abbaspourtorbati2016swiss}, or active and reactive power injections differentiated by their nodes, durations, and scheduled times. Supplies of the same type from different bidders are fungible to the central operator. We assume that each bidder~$l$ has a private true cost~function $c_l: \X_l \rightarrow \mathbb R_+$, $\X_l\subseteq\R_+^t$. We further assume that $0\in \X_l$ and~$c_l(0)=0$. This assumption holds for many electricity markets, for instance, control reserve markets and day-ahead markets that include generators' start-up costs~\cite{abbaspourtorbati2016swiss, xu2017efficient}.  Each bidder~$l$ then submits a bid function to the central operator, denoted by $b_l:\hat \X_l \rightarrow \mathbb R_+$, where $0\in \hat \X_l\subseteq\R_+^t$ and~$b_l(0)=0$.\footnote{There are markets that include shut-down costs or minimum output levels~\cite{xu2017efficient}. To address these markets, we draw attention to the properties and the results for which the assumptions, $c_l(0)=0$ and $b_l(0)=0$, are pivotal. } 

Given the bid profile $\BB=\{b_l\}_{l\in L}$, \textit{a mechanism} defines an allocation rule $x_l^*(\BB)\in \hat \X_l$ and a payment rule $p_l(\BB)\in\R$ for each bidder $l$. 
In electricity markets, the allocation rule is determined by the economic dispatch, that is, minimizing the procurement cost subject to some security constraints
\begin{equation}\label{eq:main_model}
\begin{split}
J(\BB)=&\min_{x\in \hat \X,\,y}\,\, \sum\limits_{l\in L} b_l(x_l) + d(x,y)\\
&\ \ \mathrm{s.t.}\ \ h(x,y)= 0,\, g(x,y)\leq 0.\\
\end{split}
\end{equation}
where $\hat \X=\prod_{l\in L}\hat \X_{l}$. In the case of a two-stage electricity market model,~$y\in\R^p$ may correspond to the second stage variables and the function~$d:\R^{t\rvert L\rvert}\times\R^p\rightarrow \R$ could be the second stage cost. The function~$h:\R^{t\rvert L\rvert}\times\R^{p}\rightarrow \R^{q_1}$ defines the equality constraints and the function~$g:\R^{t\rvert L\rvert}\times\R^{p}\rightarrow \R^{q_2}$ defines the inequality constraints. These constraints may correspond to the network balance constraints, and voltage and line limits in OPF problems. Alternatively, they may also correspond to procurement of the required amounts of power supplies, for instance, in the Swiss control reserve markets accepted reserves must have a deficit probability of less than 0.2\%. {Thus, problem (\ref{eq:main_model}) defines a general class of electricity market problems, including energy-reserve co-optimized markets\cite{xu2017efficient}, stochastic markets \cite{abbaspourtorbati2016swiss,conejo2010decision}, and AC-OPF problems~\cite{lavaei2012competitive}.} As a remark, if the problem \eqref{eq:main_model} is infeasible, the objective value is unbounded, $J(\BB)=\infty$.

Let the optimal solution of \eqref{eq:main_model} be denoted by $x^*(\BB)\in \hat \X$ and $y^*(\BB)\in\R^{p}$.\footnote{We assume that in case of multiple optima there is a tie-breaking rule.} We assume that the utility of bidder $l$ is linear in the payment received; $u_l(\BB)=p_l(\BB)-c_l(x^*_l(\BB))$.  A bidder whose bid is not accepted, $x_l^*(\BB)=0$, is not paid and $u_l(\BB)=0$. 
The utility of the operator $u_0(\BB)$ is defined by the total payment, namely, $u_0(\BB)=-\sum_{l\in L} p_l(\BB) - d(x^*(\BB),y^*(\BB))$.
This total payment can be an expected value when the function~$d$ is an expected second stage cost. As a remark, if the problem \eqref{eq:main_model} is infeasible, the utility of the operator is given by $u_0(\BB)=-\infty$.

There are several fundamental properties we desire for the mechanism. A mechanism is \textit{in\-dividually-rational} (IR) if bidders do not face negative utilities, $u_l(\BB)\geq 0$ for all $l\in L$. This property is also often referred to as voluntary participation or cost recovery. A mechanism is \textit{efficient} if the sum of all the utilities $\sum_{l=0}^{\lvert L\rvert} u_l(\BB)$ is maximized. From the definition of the utilities, we have $\sum_{l=0}^{\lvert L\rvert} u_l(\BB)=-\sum_{l\in L} c_l(x^*_l(\BB))- d(x^*(\BB),y^*(\BB))$. Notice that this value is maximized if we are solving for the optimal allocation of the market in~\eqref{eq:main_model} under the condition that the bidders submitted their true costs $\{c_l\}_{l\in L}$. As a result, we can attain efficiency by eliminating potential strategic manipulations.

Connected with the observation above, we say that a mechanism is \textit{dominant-strategy incentive-compatible} (DSIC) if the truthful bid profile $\mathcal C=\{c_l\}_{l\in L}$ is the dominant-strategy Nash equilibrium. In other words, every bidder finds it more profitable to bid truthfully, regardless of what others bid. However, unilateral deviations are not the only strategic manipulations we need to consider in order to ensure that the bidders reveal their true costs.

As the last desirable property, we consider immunity to collusion and shill bidding and this is the main topic of this paper. 
Bidders $K\subseteq L$ are \textit{colluders} if they obtain higher collective utility by changing their bids from $\mathcal C_K=\{c_l\}_{l\in K}$ to $\BB_K=\{b_l\}_{l\in K}$. A bidder $l$ is a \textit{shill bidder} if there exists a set $S$ and bids $\BB_S=\{b_k\}_{k\in S}$ such that the bidder~$l$ finds participating with bids $\BB_S$ more profitable than participating with a single truthful bid~$\CC_l$. Finally, by \textit{coalition-proof}, we mean that a group of bidders whose~bids are not accepted when bidding their true cost cannot profit from collusion, and no bidder can profit from using shill bids. We remark that it is not possible to achieve immunity to collusion from all sets of bidders. For instance, no mechanism can eliminate the situation where all bidders inflate their bid prices simultaneously, see also the examples in~\cite{beck2009revenue}.

Since the bidders strategize around the payment rule, the design of the payment rule plays a crucial role in attaining the aforementioned properties. In light of the discussions above, we discuss two well-studied payment rules that fail to attain some of these properties for the general class of electricity markets in~\eqref{eq:main_model}.

\subsection{The LMP mechanism}

The \textit{LMP mechanism} is adopted in markets where polytopic DC-OPF constraints and nondecreasing convex bids are considered. For simplicity in notation, assume there is a single bidder at each node of the network. Under this assumption, each bidder is supplying a one-dimensional power supply of a unique type. Then, the payment rule is $p_l(\BB)=\lambda_l^*(\BB)\, x^*_l(\BB)$, where $\lambda_l^*(\BB)\in\R$ is the Lagrange multiplier of the $l^{\text{th}}$ nodal balance equality constraint. See \cite{wu1996folk} for an exposition on the~calculation of the LMP payments from the Karush-Kuhn-Tucker conditions of DC-OPF problems.

Assume each bidder is a price-taker, in other words, considers the Lagrange multiplier of its node to be independent of its bid. Then, in addition to being IR, the LMP mechanism is DSIC.\footnote{IR requires $c_l(0)=0$ and $b_l(0)=0$ for all $l\in L$.} 
However, this economic rationale of the LMP mechanism involves a strong assumption not found in practice~\cite{joskow2001quantitative}. Under the LMP mechanism, a~bidder can maximize its utility by both inflating its bids and withholding its maximum supply~\cite{ausubel2014demand}. On the positive side,~in Section~\ref{sec:3}, we show that this mechanism is coalition-proof.

Another aspect to consider is that the economic rationale of the Lagrange multipliers follows from strong duality~\cite{bikhchandani1997competitive,lavaei2012competitive,warrington2012market}.\footnote{There are extensions of LMP through uplift payments to address the duality gap specifically arising from the unit commitment costs \cite{o2005efficient,gribik1993market}. However, the following discussions also apply to those works.} For DC-OPF problems, strong duality is implied by the convexity of the bid profile and the linearity of the constraints\cite{bertsekas1999nonlinear}. Strong duality, however, may not hold for the optimization problem~\eqref{eq:main_model}, and hence the Lagrange multipliers may not be meaningful in an economic sense. For instance, nonlinear AC-OPF constraints are known to yield a non-zero duality~gap for many practical problems\cite{lesieutre2011examining}, and sufficient conditions for zero duality gap are in general restrictive\cite{low2014convex}. 

\subsection{The VCG mechanism}
As an alternative, the \textit{VCG mechanism} is characterized by
$p_l(\BB)=b_l(x^*_l(\BB))+(h(\BB_{-l})-J(\BB))$, where $\BB_{-l}=\{b_k\}_{k\in L\setminus l}$.
The function $h(\BB_{-l})\in\R$ must be chosen carefully to ensure IR. A well-studied choice is the \textit{Clarke pivot rule} $
h(\BB_{-l})=J(\BB_{-l}),  $
where $J(\BB_{-l})$ is the optimal value of (\ref{eq:main_model}) with $x_l=0$, removing the bidder $l$ from both the objective and the constraints.\footnote{The general form is referred to as the Groves mechanism. The Clarke pivot rule generates the minimum total payment ensuring the IR property\cite{krishna1998}. IR requires $c_l(0)=0$ and $b_l(0)=0$ for all $l\in L$.} This mechanism is well-defined  under the assumption that a feasible solution exists when a bidder is removed. This is a practical assumption in electricity markets~\cite{xu2017efficient}.\footnote{For instance, this assumption holds for almost all IEEE test systems\cite{christie2000power}.} 
The~VCG mechanism has been shown to satisfy IR, efficiency, and DSIC for the market in~\eqref{eq:main_model}~\cite{karaca2017game}. This result is a generalization of the works in~\cite{vickrey1961counterspeculation,clarke1971multipart,groves1973incentives} which do not consider continuous goods, second stage cost and general constraints. 
Despite these theoretical properties, the VCG mechanism can suffer from collusion and shill bidding which can result in a loss of efficiency, see the examples in~\cite{karaca2017game}.

In our recent work, we showed that the VCG mechanism can be guaranteed to attain coalition-proofness only in restricted settings, such as polymatroid-type constraints and convex bids \cite[Theorems~4, 5]{karaca2017game}. Though the convex bid assumption may be reasonable in certain markets, the polymatroid requirement is restrictive. In particular, the polytope representing the constraints of a power system is in general not a polymatroid, and hence, in many electricity markets, the VCG mechanism suffers from collusion and shill bidding. 

Thus, the payment rules discussed above fail to attain some of the desired properties for the general class of electricity markets in~\eqref{eq:main_model}. Specifically, the LMP mechanism is not DSIC, and it may not be applicable to the general setting. On the other hand, the VCG mechanism is DSIC, but it is vulnerable to coalitional manipulations. In the next section, we focus our attention on mechanisms that attain the coalition-proofness property. 

\section{Coalition-Proof Mechanisms}\label{sec:3}
In this section, we study core-selecting mechanisms. These mechanisms are the exact class of coalition-proof mechanisms~\cite{day2008core} in the sense that a group of bidders whose bids are not accepted when bidding truthfully cannot profit from collusion, and a shill bidder cannot profit more than its VCG utility under truthful bidding, see \cite[Theorem~6]{karaca2017game}. 
We start by showing that, in addition to being coalition-proof, core-selecting mechanisms generalize the economic rationale of the LMP mechanism. Namely, they are also the exact class of mechanisms that ensure the existence of a competitive equilibrium. 
This result implies that the LMP mechanism is core-selecting, and hence coalition-proof. Since the LMP mechanism may not exist, we then define a class of core-selecting mechanisms, applicable to any electricity market modeled by~\eqref{eq:main_model}, that also approximates DSIC without the price-taking assumption. Finally, we prove the budget-balance of core-selecting mechanisms.

\subsection{Coalition-proofness via competitive equilibrium}
To address coalition-proofness, we first define the \textit{revealed utilities}, that is, the utilities with respect to the submitted bids. We then bring in the core from coalitional game theory to characterize coalition-proof mechanisms~\cite{osborne1994course}. 

The revealed utility of bidder $l$ is defined by $\bar u_l(\BB)=p_l(\BB)-b_l(x^*_l(\BB))$, and the revealed utility of the operator is the same as its utility, $\bar u_0(\BB)=-\sum_{l\in L} p_l(\BB) - d(x^*(\BB),y^*(\BB))$.  
For every $S\subseteq L$, let $J(\BB_S)$ be defined by
\begin{equation*}
\begin{split}
J(\BB_S)=&\min_{x\in \hat \X,\,y}\,\, \sum\limits_{l\in S} b_l(x_l) + d(x,y)\\
&\ \ \mathrm{s.t.}\ \ h(x,y)= 0,\, g(x,y)\leq 0,\, x_{-S}=0,\\
\end{split}
\end{equation*}
where the stacked vector $x_{-S}\in\R_+^{t(\lvert L\rvert -\lvert S\rvert)}$ is defined by omitting the subvectors from the set $S$. It is straightforward to see that this function is nonincreasing, $J(\BB_R) \leq J(\BB_S)$, for~$S\subseteq R$. Then, the \textit{core} is defined by
\begin{equation}\label{eq:coref}
\begin{split}
Core(\BB)=\Big\{\bar u\in\R\times\R^{\rvert L\rvert}_+ \,\rvert\, &\bar u_0+\sum\limits_{l\in L}\bar u_l=-J(\BB),\\
&\bar u_0+\sum\limits_{l\in S}\bar u_l\geq-J(\BB_S),\, \forall S \subset L \Big\}.
\end{split}
\end{equation}

A mechanism is said to be \textit{core-selecting} if its payments ensure that the revealed utilities lie in the core. Then, the payment rule is given by
$p_l(\BB)=b_l(x^*_l(\BB))+\bar u_{l}(\BB),$ 
where $\bar u(\BB)\in Core(\BB)$. For instance, the pay-as-bid mechanism is a core-selecting mechanism where $\bar u_l(\BB)=0$ for all $l\in L$, and $\bar u_0(\BB)=-J(\BB)$\cite{orcun2018game}. This implies that the core is nonempty and these mechanisms exist. 
Furthermore, core-selecting mechanisms are IR since the revealed utilities are restricted to the nonnegative orthant for the bidders in \eqref{eq:coref}.\footnote{If $b_l(0)\neq0$, $J$ may not be nonincreasing, pay-as-bid utilities may not lie in the core, and the core may not be nonempty. In this case, to guarantee that the core exists, individual-rationality constraints $\bar u_{-0}=(\bar u_1,\ldots,\bar u_{|L|})^\top\in\R^{\rvert L\rvert}_+$ can be removed. This new core is nonempty since it is downward open for $\bar u_{-0}$.}   

Core-selecting mechanisms are coalition-proof for the market in~\eqref{eq:main_model}~\cite{karaca2017game}. The core ensures this property since the inequality constraints in \eqref{eq:coref} restrict the revealed utilities such that they cannot be improved upon by forming coalitions. In fact, the VCG mechanism fails to attain coalition-proofness since it is in general not a core-selecting mechanism.
Our main result of this section shows that the core-selecting mechanisms offer an economic rationale similar to that of the LMP mechanism.
To state this result, we bring in tools from competitive equilibrium theory~\cite{mas1995microeconomic}. 
\begin{definition}\label{def:compeq}
	An allocation $x^*\in\R^{t|L|}_+$ and a set of price functions $\{\psi_l\}_{l\in L}$, where $\psi_l:\R_+^t\rightarrow\R$ and $\psi_l(0)=0$, constitute a competitive equilibrium if and only if the following two conditions hold
	\begin{align}
	&\text{(i)}\ x_l^* \in \argmax_{x_l\in \X_l}\, \psi_l(x_l)-c_l(x_l), \forall l\in L,\label{eq:firstcondition} \\
	&\text{(ii)}\ x^* \in \argmin_{x\in \R_+^{t|L|}} \left\{\min_{\substack{y:\, h(x,y)= 0\\ g(x,y)\leq 0 }}\,\, \sum\limits_{l\in L} \psi_l(x_l) + d(x,y) \right\}. \label{eq:secondcondition}
	\end{align}
\end{definition}

As a remark, price functions~$\{\psi_l\}_{l\in L}$ map from power supplies to prices, whereas the payments~$\{p_l\}_{l\in L}$ map from the bid profile to the payment received by each bidder.
Notice that the functions $\{\psi_l\}_{l\in L}$ can be nonlinear and bidder-dependent. As a result, this definition extends beyond the traditional Walrasian competitive equilibrium which requires a linear price for each type of supply \cite{mas1995microeconomic,bikhchandani2002package,parkes2002indirect,parkes2006iterativec}.

The central assumption of a competitive equilibrium is that participants do not anticipate their effects on the price functions. Consequently, bidders are willing to supply their allocations since by~\eqref{eq:firstcondition} these allocations maximize their utilities. Furthermore, these allocations are also optimal for the central operator by~\eqref{eq:secondcondition}. We highlight that a competitive equilibrium is not a game-theoretic solution, under either the cooperative or noncooperative approaches. Instead, it is a set of consistency conditions that models how payments would be formed from economic interactions. These conditions are considered to be a powerful benchmark in the economic~analysis~\cite{mas1995microeconomic}.

{Next, we show that for competitive equilibrium analysis we can restrict our attention to the optimal allocation of~\eqref{eq:main_model} under the truthful bid profile. To this end, the following lemma proves that a competitive equilibrium is efficient. \begin{lemma}\label{lem:effce}
	\hspace{-.15cm} If an allocation $x^*$ and price functions~$\{\psi_l\}_{l\in L}$ constitute a competitive equilibrium, then~$x^*=x^*(\CC)$.
\end{lemma}}
The proof is relegated to Appendix~\ref{app:lemef}.
We say that a mechanism ensures the existence of a competitive equilibrium if there exists a set of price functions $\{\psi_l\}_{l\in L}$ such that these price functions constitute a competitive equilibrium with $x^*(\CC)$, and $\psi_l(x^*_l(\CC))=p_l(\CC)$ for all $l\in L$. Under such mechanisms, we highlight that the condition in \eqref{eq:firstcondition} implies that if bidders treat their price functions to be independent of their bids, then they would bid truthfully since $x^*_l(\CC)$ maximizes their utility. 
 
 As an example, we revisit the DC-OPF problem with a single bidder at each node. Recall that $\lambda_l^*(\CC)\in\R$ is the Lagrange multiplier of the $l^{\text{th}}$ nodal balance equality constraint. If strong duality holds, then the LMP mechanism results in a competitive equilibrium with the allocation $x^*(\CC)$, and the set of price functions $\psi_l^{\text{LMP}}(x)=\lambda_l^*(\CC)\,x$ for all $l\in L$~\cite{schweppe2013spot,wu1996folk}. As a remark, in case there are several bidders on a node in a DC-OPF problem, these bidders would be supplying the same type of power supply to the grid. Notice that these bidders would also face the same LMP price function since these functions are formed by the Lagrange multipliers of the nodes. Hence, the LMP price functions are considered to be bidder-independent.

{We are now ready to prove that core-selecting mechanisms coincide with mechanisms that ensure the existence of a competitive equilibrium.\begin{theorem}\label{thm:lmpiscs}
	A mechanism is core-selecting if and only if it ensures the existence of a competitive equilibrium.
\end{theorem}}

The proof is relegated to Appendix~\ref{app:thm1}. As a corollary, the LMP mechanism is core-selecting, and hence it is also coalition-proof. Coalition-proofness can be another motivation for using LMP in DC-OPF problems.\footnote{As a side note, using Theorem~\ref{thm:lmpiscs}, we can prove that the LMP payments are upper bounded by the VCG payments, see Appendix~\ref{app:lmpvcg}} 

Theorem~\ref{thm:lmpiscs} shows that core-selecting mechanisms generalize the economic rationale of the LMP prices to core prices. The latter price can be nonlinear and bidder-dependent. 
As mentioned earlier, Lagrange multipliers may not constitute a competitive equilibrium since strong duality may not hold for the general class of electricity markets in~\eqref{eq:main_model}. For instance, deriving a meaningful payment mechanism to accompany nonconvex AC-OPF dispatch is an open problem~\cite{bosesome}. As core-selecting mechanisms exist even under nonconvex bids and nonconvex constraint sets, they are viable payment mechanisms for any market modeled by~\eqref{eq:main_model}.

Nevertheless, the VCG mechanism is the unique DSIC mechanism calculating the optimal allocation rule~\cite{green1979incentives,krishna1998}. Since the VCG mechanism may not be core-selecting, the DSIC property is in general violated under core-selecting mechanisms, and unilateral deviations can be profitable. Furthermore, competitive equilibrium theory relies on bidders treating their price functions as independent from their bids. This assumption does not take into account the full set of strategic behaviors. 
In the next section, we address the design of coalition-proof mechanisms to approximate DSIC without the price-taking assumption. We then address the case where the demand-side also submits bids.

\subsection{Designing coalition-proof mechanisms}\label{sec:3b}

\subsubsection{Approximating DSIC}\label{sec:3b1}
In this section, we characterize the class of core-selecting mechanisms that approximates DSIC by minimizing the sum of potential profits of each bidder from a unilateral deviation. Invoking the claim in~\cite{parkes2002achieving}, approximating DSIC provides us also with a meaningful method to achieve an approximate efficiency property.

First, we quantify the violation of the DSIC property under any core-selecting mechanism.\begin{lemma}\label{lem:lie}
	Under any core-selecting mechanism, the maximum profit of bidder~$l$ by a unilateral deviation from its truthful bid is $\bar u^{\text{VCG}}_l(\CC_l, \BB_{-l})-\bar u_l({\CC_l, \BB_{-l}})$, where $\bar u^{\text{VCG}}_l(\CC_l, \BB_{-l})=J(\BB_{-l})-J(\CC_l, \BB_{-l})$. 
\end{lemma}

The proof is relegated to Appendix~\ref{app:lemlie}. This lemma provides us with a measure for the loss of incentive-compatibility under any core-selecting mechanism.
Note that calculating the optimal deviation given in the proof of Lemma~\ref{lem:lie} requires full information of the bid profile. Moreover, attempting this optimal deviation involves a risk since bidding any amount higher would result in zero allocation and zero utility.

Next, we design an approximately DSIC core-selecting mechanism.
A mechanism is said to be \textit{maximum payment core-selecting} (MPCS) if its revealed utilities are given by,
\begin{equation}\label{eq:mpcs}
	\bar u^{\text{MPCS}}(\BB) = \argmax_{u \in \text{Core}(\BB)}\ \sum_{l\in L} u_l - \epsilon\, \norm{u_l-\bar u_l^{\text{VCG}}(\BB)}_2^2\,, 
\end{equation}
where $\epsilon$ is a small positive number. The second term in the objective of the problem~\eqref{eq:mpcs} is used as a tie-breaker. This term ensures that the optimizer is unique by picking the utilities that are nearest to the VCG utilities. See Figure~\ref{fig:core_picture} for an illustration of the revealed utilities for two bidders under different mechanisms.\footnote{By definition, under the MPCS mechanism sum of the utilities of the bidders are higher than the one under the LMP mechanism. However, we underline that every bidder may not receive an MPCS utility higher than its LMP utility. In the numerics, we provide one such instance.} We highlight that problem \eqref{eq:mpcs} is a convex quadratic program since the core is given by a set of linear equality and inequality constraints in \eqref{eq:coref}. In the numerics, we discuss its computational aspects.  

\begin{figure}[t]
	\begin{center}
		\begin{tikzpicture}[scale=0.6, every node/.style={scale=0.55}]
		\coordinate (r0) at (0,6);
		\coordinate (s0) at (6,0);
		\coordinate (si) at (0,0);
		\coordinate (s1) at (0,6);
		\draw[->] (0,0) -- (7,0) node[right] {\huge $\bar u_1$};
		\draw[->] (0,0) -- (0,6.75) node[above] {\huge $\bar u_2$};
		\filldraw[draw=black!80!blue,line width=.22mm, fill=gray!20] (r0) -- (s0) -- (si) -- (s1) -- cycle;
		\draw[scale=1, line width=.22mm, domain=0:6,smooth,variable=\x,black!80!blue] plot ({\x},{6-\x}) node[below] at(-2.75,1.5){\huge $\text{Core}$};
		\draw[scale=1, line width=.22mm, domain=0:6,dashed,variable=\x,black!40!red]  plot ({\x},{6}) node[below] at(8.4,6.9){\huge $(\bar u_1^{\text{VCG}},\,\bar u_2^{\text{VCG}})$};
		\draw[scale=1, line width=.22mm, domain=0:6,dashed,variable=\y,black!40!red]  plot ({6},{\y});
		\node[black!40!red] at (6,5.95) {\Huge\textbullet};
		\node[black!60!blue] at (3,2.95) {\Huge\textbullet};
		\node[black!] at (3.3,2.97) {\footnotesize\textbullet};
		\node[black!60!green] at (1.2,3.85) {\Huge\textbullet};
		\draw[scale=1, line width=.22mm, domain=0:3,smooth,variable=\x,black!60!blue] plot ({6-\x},{6-\x});
		\draw[scale=1, line width=.22mm, domain=0:0.312,smooth,variable=\x,black!60!blue] plot ({3.3+\x},{3.3-\x});
		\draw[scale=1, line width=.22mm, domain=-0.012:0.3,smooth,variable=\x,black!60!blue] plot ({3.6-\x},{3.0-\x});
		\draw[scale=1, line width=.22mm, domain=0:3,dashed,variable=\x,black!60!blue]  plot ({\x},{3});
		\draw[scale=1, line width=.22mm, domain=0:3,dashed,variable=\y,black!60!blue]  plot ({3},{\y});
		\draw[scale=1, line width=.22mm, domain=0:1.2,dashed,variable=\x,black!60!green]  plot ({\x},{3.9});
		\draw[scale=1, line width=.22mm, domain=0:3.9,dashed,variable=\y,black!60!green]  plot ({1.2},{\y});
		\node[black!40!red] at (-0.9,6) {\huge $\bar u_2^{\text{VCG}}$};
		\node[black!60!blue] at (-0.9,2.9) {\huge $\bar u_2^{\text{MPCS}}$};
		\node[black!60!green] at (-0.9,4.0) {\huge $\bar u_2^{\text{LMP}}$};
		\node[black!40!red] at (6,-.6) {\huge $\bar u_1^{\text{VCG}}$};
		\node[black!60!blue] at (3,-.6) {\huge $\bar u_1^{\text{MPCS}}$};
		\node[black!60!green] at (1.2,-.6) {\huge $\bar u_1^{\text{LMP}}$};
		\node[black] at (-.75,-.5) {\huge $(0,0)$};
		\draw[->,line width=.22mm] (-1.7,1.1) -- (2.1,1.1) ;
		\end{tikzpicture}
		\caption{Illustration of the revealed utilities under different mechanisms}\label{fig:core_picture}
	\end{center}
\vspace{.2cm}
\end{figure}
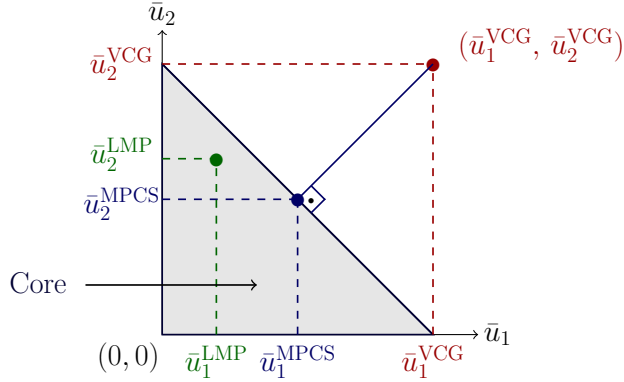

The following theorem shows that the MPCS mechanism approximates the DSIC property.

\begin{theorem}\label{thm:mpcsmax}
	The MPCS mechanism in \eqref{eq:mpcs} minimizes the sum of maximum profits of each bidder by a unilateral deviation from their truthful bids among all other core-selecting mechanisms.
\end{theorem}

The proof is relegated to Appendix~\ref{app:maxmpc}.
Under the MPCS mechanism, the total incentives to deviate from truthful bidding are minimal, and hence it is approximately DSIC.
As a remark, when the VCG utilities lie in the core, they constitute the optimizer to the problem \eqref{eq:mpcs}. This follows since $\bar u_l^{\text{VCG}}({{\BB}})=\max\{\bar u_l\ \rvert\ \bar u\in Core({{\BB}})\}$~\cite{orcun2018game}. Hence, for such instances, the MPCS mechanism is equivalent to the VCG mechanism. The LMP mechanism, however, does not have this property.

The MPCS mechanism does not rely on the price-taking assumption to approximate DSIC. In return, it generally yields nonlinear and bidder-dependent payments. Nonlinearity might be regarded as a big shift for some existing markets~\cite{bosesome}. Moreover, the bidders could find it hard to accept bidder-dependency, or this property might even be precluded by law~\cite{cramton2003electricity}. Nevertheless, the MPCS mechanism can still provide an elegant economic rationale when meaningful linear prices do not exist. The coalition-proofness property of the MPCS mechanism discourages bidders from entering the market with multiple identities to try and exploit the bidder-dependency.

\subsubsection{Market design considerations in exchanges}\label{sec:4}

In this section, we extend our model to exchange markets. 
We assume that each bidder~$l$ has a private true cost~function $c_l: \X_l \rightarrow \mathbb R$, where $0\in \X_l\subseteq\R^t$ and $c_l(0)=0$. Furthermore, the bid function is denoted by $b_l:\hat \X_l \rightarrow \mathbb R$, where $0\in \hat \X_l\subseteq\R^t$  and $b_l(0)=0$.  Note that the domains of these functions are now relaxed to $\R^t$. As a remark, an exchange market is more general than a two-sided market since these bid functions can also represent a bidder interested in buying and selling different types of supplies simultaneously.
The remaining definitions for the one-sided auctions naturally extend to exchanges. 
Moreover, results on IR, DSIC, coalition-proofness, and competitive equilibrium further hold in exchanges when we relax the power supply domains from $\R_+^t$ to $\R^t$.
For instance, the core is again defined by~\eqref{eq:coref}.\footnote{In the combinatorial exchange literature, the core is usually defined by intersecting \eqref{eq:coref} with $\bar u_0=0$~\cite{hoffman2010practical,day2013division, bichler2017core,milgrom2007package} However, this new core is in general empty\cite{hoffman2010practical}. Furthermore, such core definitions implicitly assume that a transaction can occur even without the involvement of the central operator. This is not true for the electricity markets since the dispatch has to be secure.} Then, the proof of \cite[Theorem~6]{karaca2017game} applies to exchanges, proving the connection between coalition-proofness and core-selecting.

In addition to the properties we studied so far, an exchange requires that the operator obtains a revenue, adequate to cover its total payment to balance its budget. In one-sided markets, this might be less of an issue since the demand-side is assumed to be inelastic to the price changes. We say that a mechanism is \textit{budget-balanced} if the operator has a nonnegative utility under any bid profile~$\BB$, $u_0(\BB)\geq0$. We say that it is \textit{strongly budget-balanced} if this utility is exactly zero, $u_0(\BB)=0$. 
Under DC-OPF exchange problems, the LMP mechanism is budget-balanced~\cite[Fact~4]{wu1996folk}. 
On the other hand, the VCG mechanism is not always budget-balanced.\footnote{In Appendix~\ref{app:deficvcg}, we characterize the instances in which the VCG mechanism has a deficit. In Appendix~\ref{app:coalpvcg}, we also prove that the VCG mechanism is not coalition-proof in exchanges.} This follows from the Myerson-Satterthwaite impossibility theorem, which shows that no exchange can always be efficient, DSIC, budget-balanced, and IR simultaneously~\cite{myerson1983efficient}. Fortunately, under the core selecting mechanisms, we can guarantee budget-balance.

\begin{theorem}\label{thm:corewb}
	Any core-selecting mechanism is~budget-balanced.
\end{theorem}

The proof is relegated to Appendix~\ref{app:bb}.
It follows that the MPCS mechanism is budget-balanced in addition to the properties discussed in Section~\ref{sec:3b1}. Furthermore, this also provides an alternative proof to the budget-balance of LMP.

We highlight that the LMP mechanism provides methods to reallocate its budget surplus through financial transmission and flowgate congestion rights\cite{hogan1992contract,wu1996folk}. These rights are important tools to provide market signals to incentivize investment in transmission capacity. In order to compensate the owners of these rights in the MPCS mechanism, a plausible solution is to include them as an additional constraint to \eqref{eq:mpcs}, that is, $\sum_{l \in W } \bar u_l\leq -J(\mathcal{B})-\Delta_{\text{r}}$ where $\Delta_{\text{r}}\geq 0$ is the total rights to be paid. Defining ways to incorporate these rights, and providing the correct investment signals for transmission capacity expansion, is part of our ongoing work.

\section{Numerical Results}\label{sec:5}

Our goal is to compare the effectiveness of the mechanisms we have discussed, based on electricity market examples. First, we consider the AC-OPF problem from \cite{bukhsh2013local}. This problem yields a nonzero duality gap, and hence the Lagrange multipliers cannot be guaranteed to have an economic rationale. As an alternative, we show that the MPCS mechanism coincides with the VCG mechanism, and hence it is both coalition-proof and DSIC. We then study the two-stage Swiss reserve procurement auction from~\cite{abbaspourtorbati2016swiss} which also fails to attain strong duality. For this example, the VCG mechanism is not core-selecting, and hence it is not coalition-proof. Instead, the MPCS mechanism yields a coalition-proof outcome. Finally, we consider a two-sided market with DC-OPF constraints to compare the budget-balance under the pay-as-bid, LMP, MPCS, and VCG mechanisms. We show that we can compute the MPCS mechanism in a reasonable time. All problems are solved on a computer equipped with 32 GB RAM and a 4.0 GHz quad-core Intel i7 processor. 

\subsection{AC-OPF problem with a duality gap}

The following simulation is based on a 5-bus network model given in~\cite{bukhsh2013local}. Apart from having additional generators, the network model is the same as the model in~\cite{bukhsh2013local}. We provide the generator costs for active power in Table~\ref{table:actable}. For this problem, strong duality does not hold. We verified this by showing that the semidefinite programming relaxation is not tight for this polynomial optimization problem~\cite{molzahn2014moment}. Consequently, Lagrange multipliers may not be meaningful in an economic sense. We highlight that we can solve this problem to global optimality via the second level of moment relaxations (sum-of-squares hierarchy)~\cite{molzahn2014moment}. 

Payments under the pay-as-bid and the MPCS mechanisms are provided in Table~\ref{table:actable}. Note that the pay-as-bid mechanism would actually not lead to truthful behavior. It is provided for comparison since the LMP mechanism is not applicable. For this example, even though the constraints are not polymatroid-type, the MPCS mechanism happens to coincide with the VCG mechanism, attaining the DSIC property. Moreover, we can ensure that losing bidders 3 and 4 cannot profit from collusion, and no bidder can profit from bidding with multiple identities.

\begin{table}[t]\caption{Generator data for 5-bus AC-OPF problem}\label{table:actable}

	\begin{center}
	\begin{tabular}{|l||l||l||l||l||l|}
		\hline
		Gen. & Node& Cost & $x_i^*$ MW & Pay-as-bid & MPCS\\
		\hline 
		1& 1  & $.1x_1^2+4x_1$ & 246.0 & \$7038.0& \$12772.3\\
		\hline 
		2 & 5  & $.1x_2^2+1x_2$ & 98.2 & \$1061.5 & \$2435.6 \\
		\hline 
		3 & 1  & $.1x_3^2+30x_3$ & 0 & 0 & 0 \\
		\hline 
		4 & 5  & $.1x_4^2+15x_4$ & 0 & 0 & 0\\
		\hline
	\end{tabular}
	\end{center}
\end{table}

As a remark, for this problem, we can show that there is no linear price function that would yield a competitive equilibrium. Since the bids are convex, the condition in \eqref{eq:firstcondition} requires assigning bidder~1 a linear price equal to its marginal cost at $246$MW. This price is given by $\$53.2/\text{MW}$, yielding the payment $\$13087.2$. Under this payment mechanism, utility of the bidder~1 cannot be in the core since this utility is greater than its VCG utility. This concludes the nonexistence of a competitive equilibrium in linear prices for this problem.

This AC-OPF problem was solved in 1.85 seconds using the method in~\cite{molzahn2015sparsity} with MOSEK 9 \cite{mosek2015mosek} called through MATLAB. To verify that the VCG utilities are in the core, we solved this problem under $2^2-1$ different coalitions by ignoring the core constraints that involve losing bidders, see~\cite{karaca2017game}.

\subsection{Swiss reserve procurement auctions}

The following simulation is the Swiss reserve procurement auction in the 46th week of 2014 which is based on a pay-as-bid payment rule~\cite{abbaspourtorbati2016swiss}. This auction involves 21 plants bidding for secondary reserves, 25 for positive tertiary reserves and 21 for negative tertiary reserves. The bids are discrete, that is, they are given by sets of reserve size and price pairs. Notice that the formulation in Section~\ref{sec:2} can capture such bids. The objective also includes a second stage cost corresponding to the uncertain daily auctions. Moreover, the market involves complex constraints arising from nonlinear cumulative distribution functions. These~constraints imply that the deficit of reserves cannot occur with a probability higher than 0.2\%, and they include coupling between the first and the second stage decision variables. Since this problem does not attain strong duality, meaningful linear prices cannot be derived. 
The total payments of the pay-as-bid, MPCS, and VCG mechanisms are shown in Table~\ref{tablen3}. Notice that the VCG utilities do not lie in the core since otherwise the MPCS mechanism would coincide with the VCG mechanism. As a result, the MPCS mechanism is coalition-proof, but it does not attain the DSIC property. As is discussed in Lemma~\ref{lem:lie}, we can still quantify the loss of the DSIC property by the difference between the MPCS and the VCG payments.
\begin{table}[t]

	\caption{Total payments of the reserve market (million CHF) }
	\label{tablen3}

\begin{center}
		\begin{tabular}{|l||l||l|}
			\hline
Pay-as-bid & MPCS & VCG\\
			\hline
 $2.293$ & $2.437$ & $2.529$ \\  
			\hline
		\end{tabular}
		\end{center}
	\vspace{.3cm}
\end{table}

For this market, the MPCS problem \eqref{eq:mpcs} in its current form is not computationally feasible since the core in \eqref{eq:coref} requires solutions to~\eqref{eq:main_model} under $2^{67}$ different sets of bidders. The problem can be tackled efficiently using an iterative constraint generation algorithm~\cite{day2007fair,bunz2015faster,karaca2017game}. At every iteration, the method generates the constraint with the largest violation for a provisional solution. In practice, the method requires the generation of only a few core constraints.
Using this method, the problem was solved in 8 seconds with GUROBI 7.5 \cite{gurobi} called through MATLAB via YALMIP~\cite{lofberg2005yalmip}. Computation times for the VCG mechanism and the MPCS mechanism are 580.6 and 659.2~seconds, respectively. The iterative algorithm required the generation of only 4 constraints. This shows that the MPCS mechanism can be computed in a reasonable time even when there are many market participants.

\subsection{Two-sided markets with DC-OPF constraints}

We consider the DC-OPF problem in Figure~\ref{fig:three_node}. Bids are quadratic polynomials. All lines have the same susceptance. Line limit from node $i$ to node $j$ is denoted by $C_{i,j}=C_{j,i}\in\R_+$. The optimal allocation is computed as $x^*=[0.58,\, 0.58,\, 4,\, -5.16]$ MW. The pay-as-bid mechanism yields a positive budget of~$\$48.3$. The LMP mechanism also results in a positive budget since the limits $C_{3,1}$ and $C_{3,2}$ are tight at the optimal solution~\cite[Fact~5]{wu1996folk}. This balance is $\$2.8$. For this problem, the MPCS mechanism achieves strong budget-balance with~$\$0$. Finally, the VCG mechanism has a deficit of $-\$34.8$. 

By definition, under the MPCS mechanism sum of the utilities of the bidders are higher than the one under the LMP mechanism. However, we underline that under the MPCS mechanism, not every bidder receives a utility higher than its LMP utility. In this DC-OPF problem, the total utility of the supply-side reduces by $\$4.1$, whereas the utility of the demand-side increases by $\$6.9$ when we compare the MPCS outcome with the LMP outcome. 

This problem was solved in 0.32 seconds with GUROBI 7.5 \cite{gurobi} called through MATLAB via YALMIP~\cite{lofberg2005yalmip}. To calculate the MPCS payments, we solved the problem under $2^4-1$ different coalitions. 
\begin{figure}[t]
	
	\begin{center}
		\begin{tikzpicture}[scale=0.6, every node/.style={scale=0.6}]
		\draw[-,black!80!blue,line width=.32mm] (-0.3,0) -- (-0.3,0.42);
		\draw[-,black!80!blue,line width=.32mm] (0.3,0) -- (0.3,0.42);
		\draw[-,black!80!blue,line width=.32mm] (0.3,3.58) -- (0.3,4);
		\draw[-,black!80!blue,line width=.32mm] (-0.3,3.58) -- (-0.3,4);
		\draw[-,black!80!blue,line width=.32mm] (-3.59,1.8) -- (-4,1.8);
		\draw[-,black!80!blue,line width=.32mm] (-3.59,2.2) -- (-4,2.2);
		\draw[-,black!80!blue,line width=.32mm] (3.59,1.8) -- (4,1.8);
		\draw[-,black!80!blue,line width=.32mm] (3.59,2.2) -- (4,2.2);
		\draw[-,black!80!blue,line width=.3mm] (-3.6,1.8) -- (-0.3,0.4) node[anchor=west]  at(-5.3,0.25) {\Large $C_{3,2}=2\text{ MW}$};
		\draw[-,black!80!blue,line width=.3mm] (-0.3,3.6) -- (-3.6,2.2) node[anchor=west]  at(-5.3, 3.65) {\Large $C_{3,1}=2\text{ MW}$};
		\draw[-,black!80!blue,line width=.3mm] (0.3,3.6) -- (3.6,2.2) node[anchor=west]  at(1.75,3.65) {\Large $C_{1,4}=10\text{ MW} $};
		\draw[-,black!80!blue,line width=.3mm] (0.3,0.4) -- (3.6,1.8) node[anchor=west]  at(1.75,0.25) {\Large $C_{2,4}=10\text{ MW} $};
		\draw[-,line width=.85mm] (-1,0) -- (1,0) node[anchor=west]  {\LARGE $2$};
		\draw[-,line width=.85mm] (-1,4) -- (1,4) node[anchor=west]  {\LARGE $1$};
		\draw[-,line width=.85mm] (-4,3) -- (-4,1) node[anchor=west]  {\LARGE $3$};
		\draw[-,line width=.85mm] (4,3) -- (4,1) node[anchor=west]  {\LARGE $4$};	
		\draw[<-,line width=.3mm] (0,-0.05) -- (0,-1);
		\draw[<-,line width=.3mm] (0,4.05) -- (0,5);
		\draw[<-,line width=.3mm] (-4.05,2) -- (-5,2);
		\draw[->,line width=.3mm] (4,2) -- (5,2);
		\draw (0,-1.55) circle (.55cm) node {\LARGE $G_2$} node at(3.5,-1.75) {\LARGE $c_2(x_2)= 4x_2^2 + 5x_2,$};
		\draw (0,5.55) circle (.55cm)node {\LARGE $G_1$} node at(3.5, 6.6) {\LARGE $c_1(x_1)=5x_1^2 + 4x_1,$};
		\draw (-5.55,2) circle (.55cm) node {\LARGE $G_3$}node at(-9.5,2) {\LARGE $c_3(x_3)=x_3^2 + x_3$};
		\draw (5.56,2.05) circle (.55cm)node {\LARGE $D_4$} node at(9.5, 2) {\LARGE $c_4(x_4)=x_4^2 + 20x_4$};
			\draw node at(9.5,1) {\LARGE $-8\leq x_4\leq0$};
			\draw node at(3.5,5.6) {\LARGE $x_1\geq0$};
			\draw node at(3.5,-2.75) {\LARGE $x_2\geq0$};
			\draw node at(-9.5,1) {\LARGE $x_3\geq0$};
		\end{tikzpicture}

		\caption{Two-sided DC-OPF model}\label{fig:three_node}

	\end{center}
\end{figure}
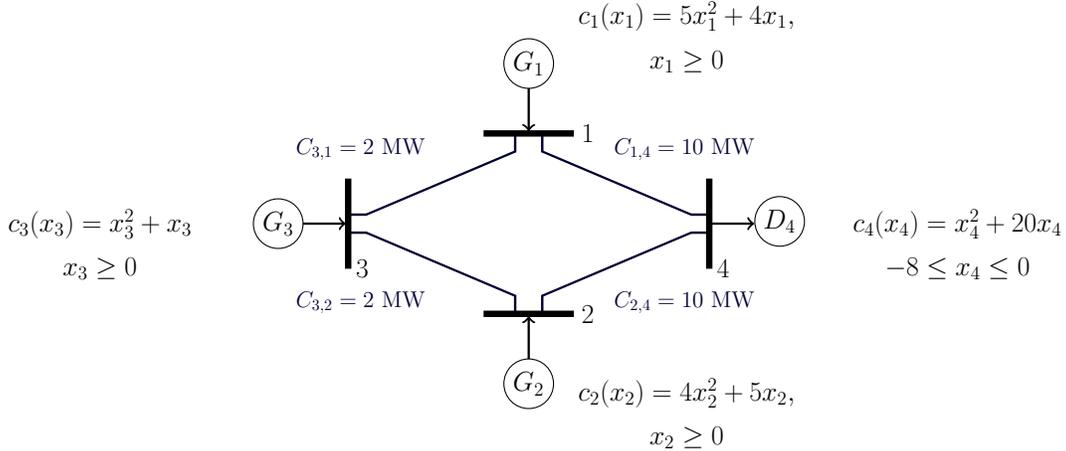

\vspace{.1cm}
\section{Conclusion}\label{sec:6}

For the general class of electricity markets, the dominant-strategy incentive-compatible VCG mechanism is susceptible to collusion and shill bidding. This motivated the design of core-selecting mechanisms for their coalition-proofness.
We showed that the well-established LMP mechanism is core-selecting, and hence coalition-proof. This result was an implication of a stronger result we proved, core-selecting mechanisms are the exact mechanisms that ensure the existence of a competitive equilibrium in linear/nonlinear prices. In contrast to LMP, we showed that core-selecting mechanisms are applicable to a broad class of markets with nonconvex bids and nonconvex constraint sets. We then characterized a class of core-selecting mechanisms that can approximate dominant-strategy incentive-compatibility without the price-taking assumption. In the case of an exchange market, we proved that core-selecting mechanisms are also budget-balanced. Our results were verified in several case studies based on realistic electricity market models. 

Our future work will explore ways to reallocate budget surplus in core-selecting mechanisms to provide correct investment~signals for transmission capacity expansion.
\vspace{.1cm}
\section*{Acknowledgments}

We are grateful to Joseph Warrington for discussions on AC-OPF, Daniel K. Molzahn for his help on the sparse moment relaxations, Jalal Kazempour for discussions on budget-balance, Sven Seuken for his feedback on the core-selecting mechanisms. We thank Swissgrid for electricity~market~data. 
\vspace{.1cm}
\appendix
\section{Proof of Lemma~\ref{lem:effce}}\label{app:lemef}
	We show that if an allocation~$x^*$ and price functions $\{\psi_l\}_{l\in L}$ constitute a competitive equilibrium, then $x^*$ is the optimal solution to the optimization problem defined by $J(\CC)$ in \eqref{eq:main_model}. By the first condition in \eqref{eq:firstcondition}, we have
	\begin{equation}\label{eq:cond1trick}
	\sum_{l\in L}\psi_l(x_l^*)-c_l(x_l^*)\geq \sum_{l\in L} \psi_l(x_l)-c_l(x_l),\,\forall x\in \X,
	\end{equation}
	where $\X=\prod_{l\in L} \X_l$. Define $(x^*,y^*)$ as the optimal solution pair to the optimization problem in \eqref{eq:secondcondition}, and $(x^*(\CC),y^*(\CC))$  as the optimal solution pair to the problem defined by $J(\CC)$. 
	Then, we obtain the following
	\begin{align*}
	\sum\limits_{l\in L} c_l(&x_l^*) + d(x^*,y^*)\\& \leq  \sum_{l\in L}\psi_l(x_l^*) - \psi_l(x_l^*(\CC)) + c_l(x_l^*(\CC))+ d(x^*,y^*)  \\
	& =  J(\CC) + \sum_{l\in L}\psi_l(x_l^*) + d(x^*,y^*) - \left(\sum_{l\in L} \psi_l(x_l^*(\CC)) + d(x^*(\CC),y^*(\CC))\right) \leq J(\CC).
	\end{align*}
	The first inequality follows from~\eqref{eq:cond1trick}. We then obtain the equality by adding and subtracting the term $d(x^*(\CC),y^*(\CC))$, and substituting $J(\CC)=\sum_{l\in L}c_l(x_l^*(\CC))+d(x^*(\CC),y^*(\CC))$. The second inequality follows since $(x^*(\CC),y^*(\CC))$ is a suboptimal feasible solution to the optimization problem in \eqref{eq:secondcondition}. We remark that $(x^*,y^*)$ is originally feasible for the problem defined by $J(\CC)$, otherwise, it would not satisfy the constraints in both \eqref{eq:firstcondition} and \eqref{eq:secondcondition}. Hence, from the last inequality, we can conclude that $(x^*,y^*)$ is in fact an optimal solution pair to the optimization problem defined by $J(\CC)$. Since previously we assumed that the optimal solution is unique according to some tie-braking rule, we obtain the desired result $x^*=x^*(\CC)$.\QEDA
\section{Proof of Theorem~\ref{thm:lmpiscs}}\label{app:thm1}
We generalize the arguments from~\cite{bikhchandani2002package,parkes2002indirect} that characterize competitive equilibria of multi-item assignment problems, to continuous goods, second stage cost, general constraints, and core-selecting mechanisms. Our proof is different from these works since it does not rely on weak duality arguments that are available to the multi-item setting with simple constraints.

Invoking Lemma~\ref{lem:effce}, for both directions of the proof we restrict our attention to the optimal allocation under truthful bids.
	
	($\impliedby$) For the market \eqref{eq:main_model}, we first prove that if a mechanism ensures the existence of a competitive equilibrium, then it is a core-selecting mechanism. To do so, we show that the revealed utilities lie in the core. 
	
	Given the bid profile $\BB=\{b_l\}_{l\in L}$, allocation $x^*(\BB)$ and price functions $\{\psi_l\}_{l\in L}$, we have;
	\begin{align}
	x^*_l(\BB) &\in \argmax_{x_l\in \hat \X_l}\, \psi_l(x_l)-b_l(x_l), \forall l\in L, \label{eq:feasubtot1} \\
	x^*(\BB) &\in \argmin_{x\in \R_+^{t|L|}} \left\{\min_{\substack{y:\, h(x,y)= 0\\ g(x,y)\leq 0 }}\, \sum\limits_{l\in L} \psi_l(x_l) + d(x,y) \right\}. \label{eq:feasubtot}
	\end{align}
	These conditions must hold because the mechanism does not know the true costs $\CC$, and it has to ensures the existence of an efficient competitive equilibrium in case the true costs are given by $\BB=\{b_l\}_{l\in L}$. Notice that the price functions depend on the bid profile, $\psi_l(x_l)=\psi_l(x_l;\BB)$, and we drop this dependence for the sake of simplicity in notation. Using $\psi_l(x^*_l(\BB))=p_l(\BB)$, the revealed utilities are defined by  $\bar u_l(\BB)=\psi_l(x^*_l(\BB))-b_l(x^*_l(\BB))$, and $\bar u_0(\BB)=-\sum_{l\in L}\psi_l(x^*_l(\BB))- d(x^*(\BB),y^*(\BB))$, where $y^*(\BB)$ is the optimal solution to \eqref{eq:feasubtot}.
	
	Next, we show that $\bar u(\BB)\in Core(\BB)$. First, observe that the individual-rationality constraints are satisfied; $\bar u_l(\BB)\geq 0$ since $0\in \hat \X_l,\, \psi_l(0)=0$, and $b_l(0)=0$. Second, we have the equality constraints in~\eqref{eq:coref}: $\sum_{l\in L} \bar u_l(\BB)+\bar u_0(\BB) = -\sum_{l\in L}b_l(x^*_l(\BB)) - d(x^*(\BB),y^*(\BB))=-J(\BB)$. Third, we show that the inequality constraints in~\eqref{eq:coref} hold, that is, \begin{equation}\label{eq:stcorein}
		-\bar u_0(\BB)\leq J(\BB_S)+\sum\limits_{l\in S}\bar u_l(\BB),\, \forall S \subset L.
	\end{equation}
	Define the following restricted problem for any subset $S\subset L$,
	\begin{equation}\label{eq:restrict}
	\begin{split}
	 \mu_0(S) =  -&\min_{\substack{x\in\hat \X,\,y\\ x_{-S}=0}}\ \sum\limits_{l\in S} \psi_l(x_l) + d(x,y)\\
	&\ \ \mathrm{  s.t. } \ h(x,y)= 0,\, g(x,y)\leq 0,
	\end{split}
	\end{equation}
	where the optimal allocation is denoted by $x^*(S)$. Because $x^*(S)$ is a feasible solution to \eqref{eq:feasubtot}, we obtain $-\bar u_0(\BB) \leq -\mu_0(S)$. 
	We then let  $x^*(\BB_S)$ be the optimal solution to $J(\BB_S)$. This solution is a suboptimal feasible solution to \eqref{eq:restrict} and hence $-\mu_0(S)\leq\sum_{l\in S} \psi_l(x^*_l(\BB_S))+d(x^*(\BB_S),y^*(\BB_S)).$ 
	Then, it suffices to show that 
	\begin{equation}\label{eq:deseq}\begin{split}
	\sum_{l\in S} \psi_l(x^*_l(\BB_S))+&d(x^*(\BB_S),y^*(\BB_S))\\ 
	&\leq J(\BB_S)+\sum_{l\in S}\bar u_l(\BB),\,\forall S\subset L,
	\end{split}\end{equation}since this would imply the inequality in \eqref{eq:stcorein}. 
	Via the condition in~\eqref{eq:feasubtot1}, we have 
	\begin{equation*}
	\psi_l(x^*_l(\BB_S))-b_l(x^*_l(\BB_S))\leq \psi_l(x^*_l(\BB))-b_l(x^*_l(\BB))=\bar u_l(\BB).
	\end{equation*}
Summing the inequality above over all $l\in S$, we obtain 
	$$\sum_{l\in S} \psi_l(x^*_l(\BB_S))-b_l(x^*_l(\BB_S))\leq\sum_{l\in S}\bar u_l(\BB).$$
By adding $d(x^*(\BB_S),y^*(\BB_S))$ on both sides and reorganizing, the above inequality yields~\eqref{eq:deseq}.	
Consequently, we have $-\bar u_0(\BB)\leq J(\BB_S)+\sum_{l\in S}\bar u_l(\BB)$, for any $S\subset L$. Hence, the revealed utilities lie in $Core(\BB)$. 
	
	($\implies$) We now prove that any core-selecting mechanism ensures the existence of an efficient competitive equilibrium. In other words, we show that, for the truthful optimal allocation $x^*(\CC)\in \X$, there exists a set of price functions $\psi_l:\R_+^t\rightarrow\R$,~$\forall l$ such that the conditions in Definition~\ref{def:compeq} are satisfied, and $\psi_l(x^*_l(\CC))=p_l(\CC)$.  Consider the utility allocation $u\in Core(\CC)$ of a core-selecting mechanism under truthful bidding. Then, define the price functions $\{\psi_l\}_{l\in L}$ as follows
	\begin{equation*}
	\psi_l(x) = 
	\begin{cases}
	0 & x=0\\
	c_l(x)+u_l& x\in \X_l \setminus \{0\}\\
	\infty& \text{otherwise}. \\
	\end{cases}
	\end{equation*}
	For these price functions, the first condition in~\eqref{eq:firstcondition} holds by construction. To show that, we study two possible cases. If $x^*_l(\CC)$ is nonzero, then $x^*_l(\CC)\in \argmax_{x_l\in \X_l}\, \psi_l(x_l)-c(x_l)$ since $u_l\geq 0$. On the other hand, if $x^*_l(\CC)=0$, then $u_l=0$ and $0\in \argmax_{x_l\in \X_l}\, 0$.

	We prove the second condition \eqref{eq:secondcondition} by contradiction. Assume there exists $x,\,y$ such that $h(x,y)=0$, $g(x,y)\leq0$, and{\begin{equation}\label{eq:contrastate}
	 \sum_{l\in L} \psi_l(x_l^*(\CC)) + d(x^*(\CC),y^*(\CC))>\sum_{l\in L} \psi_l(x_l) +d(x,y).
	\end{equation}} Define a subset $S\subseteq L$ such that $x_l=0$ for all $l\in L\setminus S$. Observe that if $x_l>0,\,\forall l$, then $S=L$. 
	Then, the core implies
	{\begin{equation}\label{eq:coretr}
	-u_0=J(\CC)+\sum\limits_{l\in L} u_l= \sum_{l\in L} \psi_l(x_l^*(\CC)) + d(x^*(\CC),y^*(\CC)),
	\end{equation}}
	The second equality follows from the definition of the price functions. Using the second equality in~\eqref{eq:coretr}, the inequality in~\eqref{eq:contrastate} is equivalent to
	\begin{equation*}
	J(\CC)+\sum\limits_{l\in L} u_l>\sum_{l\in S} \psi_l(x_l) + d(x,y).
	\end{equation*}
	By the definition of $\psi_l$, we obtain 
	\begin{equation*}
	J(\CC)+\sum\limits_{l\in L\setminus S} u_l>\sum_{l\in S} c_l(x_l) + d(x,y) \geq J(\CC_S),
	\end{equation*}
	where the last inequality is from the feasible suboptimality of $(x,y)$ for the problem defined by $J(\CC_S)$. This is because $h(x,y)=0$, $g(x,y)\leq0$, and $x\in \X$ (otherwise, the price functions are unbounded). 
	
	Using the first equality in \eqref{eq:coretr}, we have $u_0+\sum_{l\in S}u_l<-J(\CC_S)$ for $S\subseteq L$. This contradicts $u\in Core(\CC)$, and the mechanism cannot be core-selecting. We conclude that $x^*(C)$ is optimal for the central operator given the price functions. As a result, $x^*(\CC)$ and $\{\psi_l\}_{l\in L}$ constitute a competitive equilibrium. Finally, from $p_l(\CC)=c_l(x^*_l(\CC))+u_l$, we have $\psi_l(x^*_l(\CC))=p_l(\CC)$ for each bidder $l$. This concludes the proof. \QEDA
\section{Comparison of the LMP and the VCG Payments}\label{app:lmpvcg}
\begin{proposition}\label{prop:upb}
	Given any bid profile $\BB$, for every bidder~$l$ the payment under the LMP mechanism is upper bounded by the payment under the VCG mechanism.
\end{proposition}
\begin{proof}
	A similar result was proven in~\cite{xu2017efficient}, using convex analysis in the context of DC-OPF markets. We provide a simple and more general proof applicable to any setting where the LMP mechanism ensures the existence of a competitive equilibrium.
	The proof is an adaptation of \cite[Theorem~5]{ausubel2002ascending} that compares the utilities of iterative ascending auctions with that of the VCG mechanism. Since the LMP utilities lie in the core, it suffices to show that the VCG revealed utilities are greater than any other revealed utility in the core. 
	
	The VCG payment of bidder $l$ is given by $$p_l^{\text{VCG}}=b_l(x^*_l(\BB))+(J(\BB_{-l})-J(\BB)),$$  whereas the revealed VCG utility is $\bar u_l^{\text{VCG}}=J(\BB_{-l})-J(\BB)$.
	Assume there exists a revealed utility allocation $\tilde{u}\in Core(\BB)$ where $\tilde{u}_l> \bar u_l^{\text{VCG}}$. These utilities are blocked by the coalition $L_{-l}$; $$-J(\BB_{-l})> -J(\BB)-\tilde{u}_l=\tilde{u}_0+\sum_{k\in L_{-l}}\tilde{u}_k,$$ where the equality follows from the definition of the core. This contradicts that $\tilde{u}\in Core(\BB)$. We conclude that the core utilities are upper bounded by the ones under the VCG mechanism. We obtain the proposition. \QEDA
\end{proof}

Proof of Proposition~\ref{prop:upb} simplifies the arguments in~\cite{xu2017efficient} greatly. Furthermore, this upper bound is tight for some core-selecting mechanism since $\bar u_l^{\text{VCG}}=\max\left\{\bar u_l\ \rvert\ \bar u\in Core(\BB)\right\}$, see \cite[Theorem~2]{orcun2018game}. 
\section{Proof of Lemma~\ref{lem:lie}}\label{app:lemlie}
	Assume there exists a bid $\hat b_l:\hat \X_l \rightarrow \mathbb R_+$ such that 
	\begin{equation*}
	\big[\hat{b}_l(x_l^*(\hat{{\BB}}_l, \BB_{-l})) - c_l(x_l^*(\hat{{\BB}}_l, \BB_{-l})) + \bar u_l(\hat{{\BB}}_l, \BB_{-l})\big] - \bar u_l({\CC_l, \BB_{-l}}) > \big[J(\BB_{-l})-J(\CC_l, \BB_{-l})\big]-\bar u_l({\CC_l, \BB_{-l}}),
	\end{equation*}
	where $\hat{\BB}_l=\{\hat b_l\}$.
	This is equivalent to the existence of a deviation that is more profitable than the given upper bound. 
	Notice that the following holds $$\bar u_l({\hat{{\BB}}_l, \BB_{-l}})\leq \bar u^{\text{VCG}}_l(\hat{{\BB}}_l, \BB_{-l})=J(\BB_{-l})-J(\hat{{\BB}}_l, \BB_{-l}),$$ since $\bar u_l^{\text{VCG}}(\hat{{\BB}}_l, \BB_{-l})=\max\{\bar u_l\ \rvert\ \bar u\in Core(\hat{{\BB}}_l, \BB_{-l})\}$~\cite{orcun2018game}. 
	Combining the inequalities above, we have
	\begin{equation*}
	\begin{split}
	\hat{b}_l(x_l^*(\hat{{\BB}}_l, \BB_{-l})) - c_l(x_l^*(\hat{{\BB}}_l, \BB_{-l}))& + J(\BB_{-l})-J(\hat{{\BB}}_l, \BB_{-l}) \\
	&> J(\BB_{-l})-J(\CC_l, \BB_{-l}).
	\end{split}
	\end{equation*}
	We highlight that the first term is the VCG utility under a non-truthful bid, whereas the second term is the VCG utility under a truthful bid. The strict inequality above contradicts the dominant-strategy incentive-compatibility of the VCG mechanism. We conclude that $\bar u^{\text{VCG}}_l(\CC_l, \BB_{-l})-\bar u_l({\CC_l, \BB_{-l}})$ is an upper bound on the gain from a unilateral deviation. 
	
	Next, define $\epsilon$ to be a small positive number which is required to avoid ties. It is straightforward to show that the following bid achieves exactly the truthful VCG utility, and hence the exact upper bound in $\bar u^{\text{VCG}}_l(\CC_l, \BB_{-l})-\bar u_l({\CC_l, \BB_{-l}})$:
	\begin{equation*}
	\hat{b}_l(x) = 	\begin{cases}
	0 & x=0\\
	c_l(x)+\bar u^{\text{VCG}}_l(\CC_l, \BB_{-l})-\epsilon& x\in \X_l \setminus 0\\
	\infty& \text{otherwise}. \\ \end{cases}
	\end{equation*}
	Since the market solves for the optimal allocation in \eqref{eq:main_model}, if bidder~$l$ is allocated a positive quantity while bidding truthfully, then this bidder is also allocated a positive quantity while bidding $\hat{b}_l$. Moreover, under any core-selecting mechanism, we have $\epsilon\geq\bar{u}_l^{\text{VCG}}(\hat{\BB}_l, \BB_{-l})\geq\bar{u}_l(\hat{\BB}_l, \BB_{-l})\geq0$ where $\hat{\BB}_l=\hat b_l$. As a result, by bidding $\hat b_l$, the bidder~$l$ obtains its truthful VCG utility, ${u}_l(\hat{\BB}_l, \BB_{-l})=\bar u^{\text{VCG}}_l(\CC_l, \BB_{-l})-\epsilon+\bar{u}_l(\hat{\BB}_l, \BB_{-l})$, and achieves exactly the upper bound in the lemma. \QEDA
\section{Proof of Theorem~\ref{thm:mpcsmax}}\label{app:maxmpc}
	For the proof, we ignore the second term in the objective of the MPCS mechanism since it is required only for tie-breaking purposes. Assume bidders are revealing their true costs $\CC=\{c_l\}_{l\in L}$ under the MPCS mechanism. We can reformulate the problem \eqref{eq:mpcs} as follows,
	\begin{align}
	\bar u^{\text{MPCS}}(\CC) &= \argmax_{u \in \text{Core}(\CC)}\ \sum_{l\in L} u_l - \sum_{l\in L} \bar u_l^{\text{VCG}}(\CC) \nonumber \\
	&=\argmin_{u \in \text{Core}(\CC)}\ \sum_{l\in L} (\bar u_l^{\text{VCG}}(\CC)-u_l)\ . \label{eq:mpcs3}
	\end{align}
	Invoking Lemma~\ref{lem:lie}, the problem in \eqref{eq:mpcs3} implies that the MPCS mechanism minimizes the sum of maximum profits of each bidder from a unilateral deviation among all other core-selecting mechanisms.
\QEDA

\section{Characterizing Deficit Under the VCG Mechanism}\label{app:deficvcg}
Notice that the Myerson-Satterthwaite impossibility theorem does not rule out the possibility of having realizations of the VCG mechanism that are budget-balanced. 
Next, we extend the impossibility theorem by showing that in a DC-OPF market with no line limits the VCG mechanism is at most strongly budget-balanced.
\begin{proposition}\label{prop:bbvcg}
	The VCG mechanism never yields a positive utility for the operator in a DC-OPF market with no line limits.
\end{proposition}

We need the following lemma for our proof.
\begin{lemma}\label{lem:templem}
	Assume $J$ is modeled by
	\begin{equation}\label{eq:simplej}
		J(\CC)=\min_{x_l\in \X_l,\,\forall l}\, \sum_{l\in L} c_l(x_l)\ \mathrm{s.t.}\, \sum_{l\in L } x_l=0,
	\end{equation}
	where $c_l,\,\forall l$ are convex increasing and $\X_l,\,\forall l$ are polytopic constraints. Define $q\,\CC$ as the bid profile consisting of $\CC$ replicated $q$ times. Then, for any $q\in\N_+$, we have $qJ(\CC) = J(q\,\CC)$.
\end{lemma}
\begin{proof}
	We prove this by showing that the optimal solution to $J(q\,\CC)$ is given by concatenating the decision variables in the optimal solution of $J(\CC)$ $q$ times. 
	Since the problem defined by $J$ satisfies constraint qualification conditions, the KKT conditions are both necessary and sufficient for the optimality of a solution~\cite{bertsekas1999nonlinear}. We see that the KKT conditions of $J(q\,\CC)$ are satisfied by a primal solution which is the concatenation of the optimal solution of $J(\CC)$ in~\eqref{eq:simplej}~$q$~times, and a dual solution which is the Lagrange multiplier of the equality constraint of $J(\CC)$ in~\eqref{eq:simplej}. This concludes that $qJ(\CC) = J(q\,\CC)$. \QEDA
\end{proof}	

Note that this market model includes the DC-OPF markets where the network graph is connected and there are no line limits. 
We are now ready to prove Proposition~\ref{prop:bbvcg}.
\begin{proof}(Proof of Proposition~\ref{prop:bbvcg})	
We prove this by contradiction. Under the VCG mechanism, assume the operator has a positive utility:
	\begin{equation*}
	0< u_0(\CC) = - J(\CC) - \sum_{l\in W} (J(\CC_{-l})-J(\CC)), 
	\end{equation*}
	where $W\subseteq L$ is the set of bidders whose allocations are not zero. By reorganizing, we obtain
	\begin{equation*}
	|W| J(\CC)>J(\CC_W)+\sum_{l\in W} J(\CC_{-l})\geq J(|W|\,\CC) .
	\end{equation*}
	where $|W|\,\CC$ is a bid profile consisting of $\CC$ replicated $|W|$ times.
	The last inequality follows because the allocation of the problems on the left is a suboptimal feasible allocation to the problem on the right. Note further that $|W|J(\CC) = J(|W|\,\CC)$. This follows from Lemma~\ref{lem:templem}. We obtain a contradiction $J(|W|\,\CC)< J(|W|\,\CC)$. Hence, the VCG mechanism achieves at most strong budget-balance, and it never yields a positive utility for the operator in this case. \QEDA
\end{proof}

As a remark, the LMP mechanism is strongly budget-balanced for the same market under any bid profile~\cite[Fact 5]{wu1996folk}. 
Under the VCG mechanism, it is straightforward to create a two-bidder example of the market in \eqref{eq:simplej} that yields a negative utility for the operator. 

\begin{example}
	Suppose there are two bidders in the market~\eqref{eq:simplej}. The cost function of bidder 1 is given by $c_1(x_1)=x_1,\, 0\leq x_1\leq1$. The cost function of bidder 2 is given by $c_2(x_2)=3x_2,\, -1\leq x_1\leq0$. Under the VCG mechanism, bidder $1$ receives the payment $\$3$ since $p_1^{\text{VCG}} = 1 + (0-(-2))= \$3$. Whereas bidder 2 makes the payment $\$1$ since $p_2^{\text{VCG}} = -3 + (0-(-2))= -\$1$. Hence, the central operator has a $\$2$ deficit.
\end{example}
\section{Coalition-Proofness of the VCG Mechanism in an Exchange Market}\label{app:coalpvcg}
We first bring in the definition of supermodularity.
\begin{definition}\label{def:sup}
	A function $J$ is \textit{supermodular} if  
	$J(\BB_S)-J(\BB_{S\setminus{l}})\leq J(\BB_R)-J(\BB_{R\setminus{l}})$
	for all coalitions $S\subseteq R\subseteq L$ and for each bidder $l\in S$ under any bid profile. 
\end{definition}

We show that the VCG mechanism is not coalition-proof in an exchange market by showing that the market objective function~$J$ can never be supermodular.
For the following, we assume that if there is at most one bidder in the exchange, no exchange occurs, that is, $J(\BB_l)=0,$ for all $l\in L$ and $J(\emptyset)=0$.
\begin{proposition}\label{pro:cpvcg} 
	In an exchange market, the function $J$ in~\eqref{eq:main_model} is never supermodular unless no exchange occurs, $J(\BB_S)=0,\, \forall S\subseteq L$.
\end{proposition}
\begin{proof}
	Assume that the function $J$ is supermodular. Then, for bidders $i,j,k\in L$, we have
	\begin{equation*}
	J(\BB_{\{i,j\}})-J(\BB_j)\leq J(\BB_{\{i,j,k\}})-J(\BB_{\{j,k\}}).
	\end{equation*}
	This concludes that $J(\BB_{\{j,k\}})\leq 0$, since $J(\BB_j)=0$ and $J(\BB_{\{i,j\}})\geq J(\BB_{\{i,j,k\}})$. Supermodularity further implies that
	\begin{equation*}
	J(\BB_k)-J(\emptyset)\leq J(\BB_{\{j,k\}})-J(\BB_j).
	\end{equation*}
	This yields $J(\BB_{\{j,k\}})\geq 0$ since $J(\emptyset)=0$. Hence, we obtain $J(\BB_{\{j,k\}})= 0$. Moreover, this holds for any bidder pair $j,k$. By using this result, we repeat the steps above to obtain $J(\BB_{\{i,j,k\}})= 0,\, \forall i,j,k$. We conclude that similar steps can be repeated until we obtain $J(\BB_S)=0,$ for every subset $S\subseteq L$.	
	Note that this holds under any bid profile $\BB$. \QEDA
\end{proof}

Invoking \cite[Theorem~3]{karaca2017game}, we conclude that the VCG mechanism is not coalition-proof. We highlight that, to the best of our knowledge, this impossibility result for exchanges is novel. An intuition behind it is that an exchange allows for different kind of manipulations than the ones in a one-sided auction. For instance, a bidder can enter the market both as a buyer and as a seller. Then, it can manipulate the outcome by changing the amount it buys and it sells, knowing that its final allocation is given by their difference.

\section{Proof of Theorem~\ref{thm:corewb}}\label{app:bb}
	Consider the revealed utility allocation $\bar u\in Core(\BB)$ of a core-selecting mechanism. Using $\bar u_0=-J(\BB)-\sum_{l\in L}\bar u_l$, we derive an equivalent characterization of the inequality constraints in the core as follows
	$$\sum_{l\in L\setminus S}\bar u_l\leq J(\BB_S)-J(\BB),\, \forall S \subseteq L.$$
	Setting $K=L\setminus S$, these inequalities are equivalent to $\sum_{l\in K}\bar{u}_l({\BB})\leq J({\BB}_{-K} ) - J({\BB})$, for every set of bidders $K\subseteq L$. 
	By keeping the most binding constraints, we obtain
	\begin{equation}\label{eq:bind}
	\sum_{l\in K}\bar{u}_l\leq J({\BB}_{-K} ) - J({\BB}),\, \forall K\subseteq W,
	\end{equation}
	where $W\subseteq L$ is the set of bidders whose allocations are not zero.
	Next, we assume that the pay-as-bid mechanism is budget-balanced, $u_0(\BB)=-J(\BB)\geq 0$ under any bid profile. This assumption is satisfied in many exchange markets, for instance, combinatorial exchanges~\cite{parkes2002achieving}, DC-OPF exchange problems~\cite{wu1996folk} and two-sided electricity markets~\cite{xu2017efficient}.  
	Using this assumption, the inequality in \eqref{eq:bind} implies that $$\sum_{l \in W } \bar u_l\leq J(\mathcal{B}_{-W}) - J(\mathcal{B})\leq- J(\mathcal{B}),$$ since $J(\mathcal{B}_{-W})\leq0$. By reorganizing, we have $- J(\mathcal{B})-\sum_{l \in L } \bar u_l\geq 0$, since bidders who receive zero allocation are not paid. Using the equality constraint of the core in \eqref{eq:coref}, we conclude that $\bar u_0=u_0\geq 0$. Hence, the central operator ends up with a nonnegative utility. \QEDA
	
\bibliographystyle{IEEEtran}
\bibliography{IEEEabrv,library}	
\end{document}